\documentclass[11pt]{article}
\usepackage[utf8]{inputenc}
\usepackage{amssymb}
\usepackage{amsfonts,amsmath,amsthm}
\usepackage{mathtools,braket}
\usepackage[margin=1in]{geometry}
\usepackage{enumerate,enumitem}
\usepackage{multirow, multicol}
\usepackage{thmtools, thm-restate}
\usepackage[plain]{algorithm}
\usepackage{algpseudocode}
\usepackage[square,sectionbib,authoryear]{natbib} 
\usepackage{varwidth}
\usepackage{url}
\usepackage{authblk} 
\usepackage[breaklinks,hyperfootnotes=false,colorlinks=false,hidelinks]{hyperref}

\newtheorem{theorem}{Theorem}

\newtheorem{proposition}[theorem]{Proposition}
\newtheorem{corollary}[theorem]{Corollary}
\newtheorem{example}{Example}

\newtheorem{claim}{Claim}
\newtheorem{definition}{Definition}

\algnewcommand\algprocedure{\textbf{Procedure:}}
\algnewcommand\Procedurename{\item[\underline{\algprocedure}]}
\algnewcommand\algmain{\textbf{Main:}}
\algnewcommand\Main{\item[\underline{\algmain}]}
\algnewcommand\algorithmicinput{\textbf{Input:}}
\algnewcommand\Input{\item[\algorithmicinput]}
\algnewcommand\algorithmicoutput{\textbf{Output:}}
\algnewcommand\Output{\item[\algorithmicoutput]}

\urldef{\email}\path|{vijay.menon, kate.larson}@uwaterloo.ca|

\makeatletter
\def\@lbibitem[#1]#2{%
\if\relax\@extra@b@citeb\relax\else
\@ifundefined{br@#2\@extra@b@citeb}{}{%
\@namedef{br@#2}{\@nameuse{br@#2\@extra@b@citeb}}}\fi
\@ifundefined{b@#2\@extra@b@citeb}{\def\NAT@num{}}{\NAT@parse{#2}}%
\item[\hfil\hyper@natanchorstart{#2\@extra@b@citeb}\citep{#2}%
\hyper@natanchorend]%
\NAT@ifcmd#1(@)(@)\@nil{#2}}
\makeatother

\makeatletter
\renewcommand\@biblabel[1]{#1} 

\renewcommand\newblock{\hskip .11em\@plus.33em\@minus.07em}
\makeatother

\begin{document}
 \title{Deterministic, Strategyproof, and Fair Cake Cutting}

\author[]{Vijay Menon}
\author[]{Kate Larson}
\affil[]{David R. Cheriton School of Computer Science\\ University of Waterloo \\ \email}
\date{}

\maketitle

\begin{abstract}
 We study the classic cake cutting problem from a mechanism design perspective, in particular focusing on deterministic mechanisms that are strategyproof 
and fair. We begin by looking at mechanisms that are non-wasteful and primarily show that for even the restricted class of piecewise constant valuations there exists no 
direct-revelation mechanism that is strategyproof and even approximately proportional. Subsequently, we remove the non-wasteful constraint and show another impossibility 
result stating that there is no strategyproof and approximately proportional direct-revelation mechanism that outputs contiguous allocations, again, for even the 
restricted class of piecewise constant valuations. In addition to the above results, we also present some negative results when considering an approximate notion of 
strategyproofness, show a connection between direct-revelation mechanisms and mechanisms in the Robertson-Webb model when agents have piecewise constant valuations, and 
finally also present a (minor) modification to the well-known Even-Paz algorithm that has better incentive-compatible properties for the cases when there are two or three 
agents.  
\end{abstract}
\section{Introduction}

Imagine a scenario where there is heterogeneous divisible good that is to be divided among a certain set of $n$ agents. For an appropriately chosen  notion of 
fairness, ideally, we would like this resource to be divided fairly among these $n$ agents and so a natural question that arises is on how one would accomplish this. The 
cake-cutting problem metaphorically refers to this very scenario and it represents a fundamental problem in the theory of fair division. More formally, in the cake 
cutting problem, the cake is modelled as the interval $[0, 1]$ and each of the $n$ agents is assumed to have a valuation function over the cake. The goal, as described above, 
is to partition the cake so that it is fair according to some chosen notion of fairness. Starting with the work of \citet{steinhaus48}, it has been studied over the last 
several decades by mathematicians, economists, and political scientists  (see e.g.\ the books by \citet{brams96} and \citet{rw98}), and more recently has attracted the 
attention of computer scientists (see e.g.\ the survey by \citet{pro16}) partly due to the nature of the challenges involved---that often require an algorithm design or 
complexity point of view---and partly due to its relevance in the design of multiagent systems \citep{cheval06}. 

Most work in the cake cutting literature focuses on addressing the above described scenario and so most of it is concerned with computing a fair allocation using 
{proportionality} or {envy-freeness} as the main fairness criterion (see \citep{aziz16} and also \citep{pro16} for a recent survey). While addressing the above scenario is 
indeed the core part of the 
problem, the fact that the valuation functions of the agents may be their private information entails that the designed protocols or mechanisms be strategyproof---meaning 
that there is no incentive to misreport one's valuation function---for if otherwise the agents can lie about their valuation functions and potentially benefit from it. 
Therefore, the focus of this paper, like in \citep{chen10, chen13, mossel10, maya12, aziz14, branzei15}, is to look at mechanisms that are not only fair, but also 
strategyproof. 

In this paper, we look at deterministic mechanisms and we are primarily focused on the direct-revelation model---one where the agents reveal their entire valuation function. 
Our main objective is to better understand the limits of determinism when it comes to imposing strategyproofness and fairness (or approximate notions of either of 
them)---i.e., to see what is or is not achievable with deterministic mechanisms---and to this end we make the following contributions. 
\begin{enumerate}[label={\alph*})]
  \item In Section~\ref{nwm} we begin by looking at non-wasteful mechanisms and we show the following strong impossibility which in turn strengthens an impossibility result 
given \citet[Theorem 7]{aziz14}. 
  \begin{enumerate}[label={\arabic*}.]
 \item \textit{For any $n \geq 2$ and $0 \leq \epsilon_1, \epsilon_2 < \frac{1}{n}$ such that $3\epsilon_1 + \epsilon_2 < \frac{1}{n}$, there is no deterministic and 
non-wasteful mechanism for $n$ agents with piecewise constant valuations that is $\epsilon_1$-strategyproof and $\epsilon_2$-proportional.}
\end{enumerate}

\item In Section~\ref{wm} we remove the non-wasteful constraint and provide the following results. The first one addresses a special case (one where the 
allocations are restricted to contiguous allocations) of an open question posed by \citet{chen10, chen13}. The second result establishes a connection between direct-revelation 
mechanisms and mechanisms in the Robertson-Webb model for the case when agents have piecewise constant valuations. And finally, the third one is a positive result where we 
present a (minor) modification to the well-known Even-Paz algorithm that has better incentive-compatible properties for the cases when there are two or three agents.  
    
\begin{enumerate}[label={\arabic*}., start=2]
  \setcounter{enumi}{5}
 \item  \textit{For any $n \geq 2$ and $0 \leq \epsilon < \frac{1}{n}$, there is no deterministic mechanism for $n$ agents with piecewise constant valuations that makes 
contiguous allocations and is strategyproof and $\epsilon$-proportional.}

 \item \textit{Let $P_1, P_2, P_3$ denote the property of strategyproofness, proportionality, and envy-freeness, respectively, and let $P \subseteq \{P_1, P_2, P_3\}$. 
Given 
$k$, an upper bound on the maximum number of breakpoints in any agent's valuation function, and any $n \geq 2$, if there exists a (finite) deterministic direct-revelation 
mechanism for $n$ agents with piecewise constant valuations that satisfies $P_i$, $\forall P_i \in P$, then, for any $\epsilon > 0$, there exists a mechanism in the 
Robertson-Webb model for $n$ agents with piecewise constant valuations that asks at most $n \left\lfloor \frac{2k}{\epsilon} \right\rfloor $ queries on each input and is 
$\epsilon$-$P_i$, $\forall P_i \in P$.}

 \item  \textit{There exists a mechanism that is deterministic, proportional, and $\epsilon$-strategyproof, where $\epsilon = 1 - \frac{3}{2n}$ or $\epsilon = 1 - 
\frac{3}{2n} + \frac{1}{2n^2}$ depending on whether $n$ is even or odd, respectively. This mechanism achieves a better worst-case bound (as compared to the Even-Paz 
algorithm) in terms of $\epsilon$-strategyproofness for the cases when there are two or three agents. 
}
\end{enumerate}
\end{enumerate}
\subsection{Related Work}
The cake cutting problem has been studied using two input models. The first, and in fact the most ubiquitously used model in the literature, is the Robertson-Webb model where 
the mechanism proceeds through a sequence of queries on the valuation functions of the agents. The second is the direct-revelation model where the input to the mechanism is 
the complete valuation functions of the agents. Given this, we can classify the work on strategic aspects of cake cutting into two classes based on the input model used and 
look at each of them separately. (Note that there is a large body of literature on cake cutting that does not consider strategic issues. For this, we refer the 
reader to a very recent survey by \citet{pro16}.)

There are mainly five related papers \allowbreak \citep{chen10,chen13,mossel10,maya12,aziz14,li15,ali17} that operate in the direct-revelation model and at the same time look 
at strategic issues in cake cutting. Among them, the ones that are most relevant to results in this paper are the works of \citet{chen10, chen13} and \citet{aziz14}. 
\citet{chen10, chen13} were the first to look at strategyproof cake cutting. In particular, they considered piecewise uniform, piecewise constant, and 
piecewise linear valuations, and their main result was a deterministic, strategyproof, envy-free, proportional, and Pareto-optimal mechanism for the case when the agents have 
piecewise uniform valuations. In addition to their main result, they also presented randomized algorithms that are truthful in expectation, proportional, and envy free for 
piecewise linear valuations. One of our results here (Theorem~\ref{SP-prop-cont}) addresses a special case of an open question they had posed in their paper regarding 
deterministic and strategyproof mechanisms for the case when the agents have piecewise constant valuations. 

The other most relevant work that operates in the direct-revelation model is that of \citet{aziz14}. They present two mechanisms, CCEA and MEA, for piecewise 
constant valuations, where CCEA is deterministic, robust envy-free\footnotemark, and non-wasteful, and MEA is deterministic, Pareto-optimal, and envy-free. Although 
both of them generalize the mechanism proposed by \citet{chen10, chen13} when restricted to piecewise uniform valuations, they do not remain strategyproof for piecewise 
constant valuations. Additionally, they also showed that there is no mechanism that is strategyproof, robust proportional\footnotemark[1], and non-wasteful for piecewise 
constant valuations. One of our results here (Theorem~\ref{nwImp}) strengthens the latter impossibility result.  

\footnotetext[1]{Robust envy-freeness and robust proportionality are much stronger notions than envy-freeness and proportionality, respectively; see \citep{aziz14} for 
more details.}

Besides the two papers mentioned above, \citet{mossel10} provided an impossibility result regarding the existence of deterministic mechanisms that guarantee $\frac{1}{n}$ 
when all the agents have the same valuation function and greater than $\frac{1}{n}$ otherwise. They also showed the existence of a randomized mechanism that is strategyproof 
in expectation and is proportional and envy-free (this existence result was also shown by \citet{chen10}). The work of \citet{maya12} considers the goals of strategyproofness 
and Pareto efficiency for two agents with piecewise uniform valuations and their main result shows that any strategyproof mechanism achieves at most $\approx 0.93$ fraction 
of the optimal welfare. \citet{li15} study the problem of designing strategyproof cake-cutting mechanisms in the presence of externalities when the agents have 
piecewise uniform valuations. They propose a way to model externalities and show some results on the existence of strategyproof mechanisms. Finally, \citet{ali17} study the 
problem of designing strategyproof and envy-free mechanisms with the goal of minimizing the number of pieces in the final allocation. They primarily focus on the case when the 
agents' valuations are constant over a single interval and they provide a mechanism that is strategyproof, envy-free, and cuts the cake into at most $2n-1$ pieces for this 
setting. 

Among papers that use the Robertson-Webb model, the two that are most related to this paper are the works of \citet{kurokawa13} and \citet{branzei15}. 
\citet{kurokawa13} looked at the problem of envy-free cake cutting and their main result was a parameterized protocol for piecewise linear valuations. Additionally, they also 
showed that there is no mechanism of complexity bounded only by a function of the number of agents and the total number of breakpoints that is strategyproof and 
envy-free for piecewise constant valuations.

The other most relevant paper that uses the Robertson-Webb model is that of \citet{branzei15}. \citet{branzei15} showed that any deterministic and strategyproof protocol 
is a dictatorship when there are only two hungry agents and for the case when there are more than two hungry agents they show that there is at least one agent who gets an 
empty piece. Additionally, they also provided a randomized mechanism that is strategyproof in expectation and $\epsilon$-proportional for any $\epsilon > 0$. Although our 
impossibility results here are not quite as strong as that of \citeauthor{branzei15}'s, we believe that it complements the same by considering a much more permissive input 
model.

Finally, we also briefly remark on the predefined divisible goods setting which has been extensively studied in fair division (see, for instance, \cite{cheung16} and the 
references therein). In the predefined goods setting, the task is to fairly allocate $m$ divisible items among $n$ agents, where the private information of an agent is the 
value she has for each of the items. So, in effect, the predefined goods setting can be considered as a special case of the setting of cake cutting with piecewise 
constant valuations---i.e., one where the valuation functions of all the agents have the same set of breakpoints---and therefore all the negative results in that setting 
carry over to the cake cutting setting. However, some of the nice positive results with regards to strategyproof mechanisms in this setting (see, 
e.g., \cite{cg13} and \cite{cheung16}) do not carry over to the cake cutting setting as here, informally, the difficulty in achieving strategyproofness arises from the fact 
that a mechanism has to guard against two types of manipulations: i) with regards to the breakpoints in the valuation functions and ii) with regards to the values assigned 
to the pieces between two consecutive breakpoints.

\section{Preliminaries} \label{prelims}

The cake---which is a heterogeneous divisible good---is modelled as the unit interval $[0, 1]$. A piece of cake $X$ is a finite union of disjoint (except at 
the boundaries of the intervals) subintervals of $[0,1]$. An interval is denoted by $I$ and the length of an interval $I = [x, y]$ is given by $| I | = y - x$. In any cake 
cutting instance, there are $n$ agents who want a share of the cake and we denote them by $[n] = \{1,\cdots, n\}$. Each agent $i \in [n]$ has a non-negative, integrable, and 
private value density function (also referred to as their utility/valuation function), $v_i: [0,1] \to \mathbb{R^+} \cup \{0\} $, that indicates how agent $i$ values 
different parts of the cake, and given a piece of cake $X$, the value for $X$ is defined by $V_i(X) = \int_{X} v_i(x) \,dx = \sum_{I \in X} \int_{x \in I} v_i(x) 
\,dx$. (Note that the fact that the valuation functions are integrable implies additivity---meaning for two disjoint intervals $I$ and $I'$, $V_i(I \cup I') = V_i(I) + 
V_i(I')$.) Without loss of generality, we assume that the valuation functions are \textit{normalized}, i.e., $V_i([0, 1]) = 1$ for all $i \in [n]$. Additionally, we say that 
agent $i$ is \textit{hungry} if $v_i(x) > 0$ for all $x \in [0,1]$.

\subsection{Input models}

The way we have defined the valuation functions above, it is not necessary that they have a discrete representation. Therefore, most of the work in cake cutting assumes a 
query model---commonly known as the Robertson-Webb query model---that allows only the following two types of queries. 
\begin{enumerate}
 \item Eval query: given an interval $[x, y]$, \textbf{eval}$(i, x, y)$ asks agent $i$ for its value for $[x, y]$, i.e., \textbf{eval}$(i, x, y)$ = $V_i(x, y)$.  
 \item Cut query: given a point $x$ and $r \in [0, 1]$, \textbf{cut}$(i, x, r)$ asks agent $i$ for the minimum (i.e., the leftmost) point $y$ such that $V_i (x, y) = r$. 
\end{enumerate}

Here we note that there exists an alternative formalization of the Robertson-Webb model where essentially it is required that all the cut points be defined by the agents and 
not by the center (see, e.g., \citep{woe07} and \citep{branzei15} where this formalization is used). However, here we use the slightly more permissive 
formalization as given by \citet{pro13,pro16} and as defined above, where (if necessary) the center can make additional cuts once it has finished asking all the 
queries. 

While the Robertson-Webb model is the most widely studied, there is recent work (e.g., \cite{chen10, chen13}, \cite{bei12}, \cite{aziz14}) which restricts the agents' 
valuation functions to ones that have a concise representation, thus in turn giving rise to another model (which we will refer to as the direct-revelation model) where the 
entire valuation functions of the agents are given as input to a mechanism. In this paper, we consider one fundamental class of restricted valuations, namely, piecewise 
constant valuations. An agent $i$ is said to have a piecewise constant valuation if $[0, 1]$ can be partitioned into a finite number of intervals such that $v_i$ is a 
constant over each interval (i.e., it is a step function). A special case of piecewise constant valuations are the piecewise uniform valuations where, for some constant $c$, 
the function can only take the values $0$ or $c$. One of the reasons why piecewise constant valuations are interesting is because they are expressive enough to be able to 
approximate any general valuation function. 

In this paper, we make use of both input models. All our negative (impossibility) results are in the direct-revelation model, while the algorithm we present is in the 
Robertson-Webb model. Note that the existence of a mechanism in the Robertson-Webb model implies an existence in the direct-revelation model, but the converse is not 
necessarily true, and hence our negative results, in particular, are strong. Also, all our negative results are true for even the restricted class of piecewise constant 
valuations and hence naturally carry over to anything more general (like piecewise linear valuations).  

\subsection{Properties of cake cutting mechanisms}
A direct-revelation cake cutting mechanism $\mathcal{M}$ takes the valuation functions $(v_1, \cdots, v_n)$ of the agents (which we refer to as a profile) as input and it 
outputs an allocation $A = (A_1, \allowbreak \cdots, A_n)$, where $A_i$ is the allocation to agent $i$, and $\forall i, j (\neq i)$, $A_i$ and $A_j$ are disjoint (except at 
the boundaries of the intervals). While in general $A_i$ can be any arbitrary piece of cake, sometimes we concern ourselves only with \textit{contiguous allocations} where 
each $A_i$ is a single interval. Throughout, we assume that the mechanism allocates (to some agent) all the pieces of cake that are valued at greater than zero by at least one 
of the agents. Additionally, we often use $\mathcal{M}_{i} (v_i, v_{-i})$ to denote $A_i$, where $v_{-i}$ denotes the valuation functions of all the agents except $i$. 

A mechanism is said to be \textit{non-wasteful} if it allocates every piece that is desired---meaning that the piece is assigned a value greater than zero---by at 
least one agent to an agent who desires it. More formally, if $I$ is a subinterval of $[0, 1]$ and $D(I) = \{i \in [n] \mid V_i(I) > 0 \}$, then a mechanism is non-wasteful if 
it always makes an allocation $A$ such that $\forall I, I \subseteq \cup_{i \in D(I)} A_i$. In this paper, we consider both non-wasteful and wasteful mechanisms. Additionally, 
for all our negative results we assume \textit{free disposal}, meaning that the mechanism can throw away pieces of cake that aren't valued (at greater than zero) by any of 
the agents without incurring a cost. Note that the existence of a mechanism without the free disposal assumption implies the existence of a mechanism with the free disposal 
assumption, but the converse is not necessarily true. Hence, in the context of negative results, such an assumption only makes them stronger. 

In this paper, we use (approximate) proportionality as our main fairness criterion. In the definition below, proportionality refers to the special case when $\epsilon = 0$. 
And so, throughout, when we say that a mechanism is proportional it implies that it is $0$-proportional. 

\begin{definition}[$\epsilon$-proportionality]
 For any $\epsilon \in \left[0, \frac{1}{n}\right)$, a mechanism $\mathcal{M}$ is said to be $\epsilon$-proportional if it always returns an allocation $(A_1, 
\cdots, 
A_n)$ such that $\forall i \in [n], V_i (A_i) \geq \frac{1}{n} - \epsilon$. 
\end{definition}

Along with (approximate) proportionality, the other major property we are concerned with is (approximate) strategyproofness. Informally, for any $\epsilon \in [0, 1]$, a 
cake cutting mechanism $\mathcal{M}$ is said to be $\epsilon$-strategyproof if an agent can gain a utility of at most $\epsilon$ by misreporting her valuation 
function. More formally, we have the following definition.

\begin{definition}[$\epsilon$-strategyproofness]
For any $\epsilon \in [0, 1]$, a mechanism $\mathcal{M}$ is said to be $\epsilon$-strategyproof if for every agent $i \in [n]$, and for all $v'_i, v_{-i}$, 
$V_i(\mathcal{M}_i(v_i, v_{-i})) \geq V_i(\mathcal{M}_i(v'_i, v_{-i})) - \epsilon$.  
\end{definition}

Strategyproofness refers to the special case in the above definition when $\epsilon = 0$. And so, throughout, when we say that a mechanism is strategyproof (or truthful) it 
implies that it is $0$-strategyproof, or, in other words, that for every agent truth telling is a (weakly) dominant strategy. Here we note that this is not the first time the 
notion of approximate strategyproofness has been considered. For instance, \citet{birrell11} considered it in context of voting and \citet{schummer04} considered it in the 
context of two agent two-good exchange economies. We refer the reader to a paper by \citet{lubin12} for a detailed discussion on approximate-strategyproofness. 

In addition to the above mentioned properties, we also talk about (approximate) envy-freeness and Pareto-efficiency (mainly in the context of results pertaining to previous 
work). For an $\epsilon \in [0, 1]$, a cake cutting mechanism is said to be $\epsilon$-\textit{envy-free} if it always returns an allocation $(A_1, \allowbreak \cdots, A_n)$ 
such that $\forall i, j \in [n], V_i(A_i) \geq V_i(A_j) - \epsilon$. Envy-freeness refers to the special case in the above definition when $\epsilon = 0$. A cake cutting 
mechanism is said to be \textit{Pareto-efficient} if it always returns an allocation $(A_1, \allowbreak \cdots, A_n)$ that is not Pareto-dominated---meaning that there is no 
other allocation $(A'_1, \cdots, A'_n)$ such that $\forall i \in [n], V_i(A'_i) \geq V_i(A_i)$, with the inequality being strict for at least one agent. Note that 
non-wastefulness, as defined above, can be considered as a weak form of efficiency where essentially all it is doing is restricting a mechanism from allocating a piece of zero 
value to an agent if there is some other agent who has a non-zero value for it. 

\section{Non-wasteful mechanisms} \label{nwm}

\citet{chen10,chen13} were the first to consider direct-revelation mechanisms for cake cutting and they proposed a polynomial-time deterministic mechanism that is 
strategyproof, 
proportional, envy-free, and Pareto-efficient for piecewise uniform valuations. In light of such a result, one natural question that arises is if there exists a mechanism for 
at least the more general class of piecewise constant valuations that is strategyproof, Pareto-efficient, and fair (for some notion of fairness). We already have an answer to 
this question due to \citet{sch96} who considered the predefined divisible goods setting and showed that the only mechanism that is strategyproof and Pareto-efficient 
is a dictatorship. And since this predefined goods setting is a special case of the setting with piecewise constant valuations, we know that there is no hope for achieving 
strategyproofness, Pareto-efficiency, and any notion of fairness. Therefore, in this context, another notion to consider is non-wastefulness which, as defined in 
Section~\ref{prelims}, is a weaker notion of efficiency. In fact, \citet{aziz14} considered this notion and showed an impossibility result saying that there is no 
strategyproof, robust proportional (which is a much stronger notion of proportionality; see \citep{aziz14}), and non-wasteful mechanism \citep[Theorem 7]{aziz14}. Below, we 
show a stronger impossibility which says that there is no deterministic and non-wasteful mechanism that is even $\epsilon_1$-strategyproof and $\epsilon_2$-proportional for 
$\epsilon_1, \epsilon_2$ such that $0\leq 3\epsilon_1 + \epsilon_2 < \frac{1}{n}$.  

\begin{theorem} \label{nwImp}
 For any $n \geq 2$ and $0 \leq \epsilon_1, \epsilon_2 < \frac{1}{n}$ such that $3\epsilon_1 + \epsilon_2 < \frac{1}{n}$, there is no deterministic and non-wasteful 
mechanism for $n$ agents with piecewise constant valuations that is $\epsilon_1$-strategyproof and $\epsilon_2$-proportional.
\end{theorem}

\begin{proof}
Let us assume that there exists a deterministic, $\epsilon_1$-strategyproof, $\epsilon_2$-proportional, and non-wasteful mechanism $\mathcal{M}$ for $n$ agents. First, 
consider the profile $(u, u, y, \cdots, y)$, where
\begin{equation*}
 \begin{aligned}[c]
u(x) = \begin{cases} \frac{n}{2} & x \in [0, \frac{2}{n}]\\ 0 & x \in (\frac{2}{n}, 1] \end{cases}
 \end{aligned}
\qquad
  \begin{aligned}[c]
y(x) = \begin{cases} 0 & x \in [0, \frac{2}{n}]\\ \frac{n}{n-2} & x \in (\frac{2}{n}, 1]. \end{cases}
 \end{aligned}
\end{equation*}
Let $\mathcal{M}$ make an allocation $A_1$, $A_2$ to agent 1 and agent 2, respectively. Without loss of generality, we can assume that $|A_1| \geq |A_2|$. Now, since 
$\mathcal{M}$ is non-wasteful, $|A_1| + |A_2| = \frac{2}{n}$, and so this implies that $|A_1| \geq \frac{1}{n}$ (since both cannot be less than $\frac{1}{n}$). 

Next, let us consider the profile $(v, u, y, \cdots, y)$, where
\begin{equation*}
 \begin{aligned}[c]
v(x) = \begin{cases} 
      \frac{1}{|A_1|} & x \in A_1 \\
      0 & \text{otherwise.}
   \end{cases}
\end{aligned}
\end{equation*}

Since $\mathcal{M}$ is non-wasteful and $\epsilon_1$-strategyproof, agent 1 has to get a utility of at least $1-\epsilon_1$ from the allocation in this profile, because if it 
allocates anything less, then agent 1 will deviate to Profile 1. Therefore, if $\mathcal{M}$ allocates $A'_1$ and $A'_2$ to agent 1 and 2, respectively, then 
${(1-\epsilon_1)} \cdot |A_1| \leq |A'_1|$ and so $ |A'_2| \leq |A_2| + \epsilon_1 \cdot |A_1|$.

Finally, consider the profile $(v, w, y, \cdots, y)$, where 
\begin{equation*}
\begin{aligned}[c]
 w(x) = \begin{cases} 
      \frac{1 - \delta}{|A_1|} & x \in A_1 \\
      \frac{\delta}{|A_2|} & x \in A_2\\
      0  & \text{otherwise}
   \end{cases}
  \end{aligned}
\end{equation*}

Now, here, since $\mathcal{M}$ is non-wasteful and because agent 1 does not desire any part of $A_2$, agent 2 gets at least $A_2$. However, agent 2 has a total utility of 
only 
$\delta$ for $A_2$, and since we assumed that $\mathcal{M}$ is also $\epsilon_2$-proportional, it has to get a piece of length, say, $\ell$, of $A_1$, where
\begin{align*}
 &\delta + \frac{\ell}{|A_1|}(1-\delta) \geq \frac{1}{n} - \epsilon_2\\
 &\implies \ell \geq |A_1| \cdot \frac{(\frac{1}{n} - \epsilon_2 - \delta)}{(1-\delta)}.
\end{align*}

If $\mathcal{M}$ allocates $A''_1$ and $A''_2$ to agent 1 and 2, respectively, in this profile, then $|A''_2| \geq |A_2| + \ell \geq |A_2| + |A_1| \cdot 
\frac{(\frac{1}{n} - \epsilon_2 - \delta)}{(1-\delta)}$. Also, from the profile $(v, u, y, \cdots, y)$ we know that $|A'_2| \leq |A_2| + \epsilon_1 \cdot |A_1|$. Therefore, 
if $\frac{n}{2} \left( |A_2| + |A_1| \cdot \frac{(\frac{1}{n} - \epsilon_2 - \delta)}{(1-\delta)}\right) > \frac{n}{2}\left( |A_2| + \epsilon_1 \cdot |A_1| \right) + 
\epsilon_1$, then agent 2 can deviate from $(v, u, y, \cdots, y)$ to $(v, w, y, \cdots, y)$ and as a result gain strictly more than $\epsilon_1$. So, considering this 
inequality, we have,
\begin{align*}
\frac{n}{2} \left(|A_2| + |A_1| \cdot \frac{(\frac{1}{n} - \epsilon_2 - \delta)}{(1-\delta)}\right) &> \frac{n}{2}\left( |A_2| + \epsilon_1 \cdot |A_1| \right) + \epsilon_1 \\
|A_2| + |A_1| \cdot \frac{(\frac{1}{n} - \epsilon_2 -  \delta)}{(1-\delta)} &> |A_2| + \epsilon_1 \cdot |A_1| + \frac{2}{n} \epsilon_1 \\
\frac{1}{n} \cdot \left(\frac{(\frac{1}{n} - \epsilon_2 - \delta)}{(1-\delta)} - \epsilon_1 \right) &> \frac{2}{n} \epsilon_1  && (\text{as $|A_1| \geq \frac{1}{n}$})\\
\delta &< \frac{1- n(3\epsilon_1 + \epsilon_2)}{n(1 - 3\epsilon_1)}.
\end{align*}

This in turn implies that for every $0 \leq \epsilon_1, \epsilon_2 < \frac{1}{n}$ such that $3\epsilon_1 + \epsilon_2 < \frac{1}{n}$, we can always find a $0 < \delta < 
\frac{1- n(3\epsilon_1 + \epsilon_2)}{n(1 - 3\epsilon_1)}$ to fit in the definition of the function $w(x)$ such that the inequality---and hence our theorem---will be true. 
\end{proof}

With respect to the above theorem, we highlight the following corollaries which are basically special cases when $\epsilon_1$ and $\epsilon_2$ are zero, respectively. The 
first one rules out the possibility of guaranteeing some value (however small) to all the agents if the mechanism has to be non-wasteful and strategyproof. The second 
one, on the other hand, provides a lower bound for the natural question on whether we can come up with mechanisms that can guarantee proportionality and non-wastefulness and 
at the same time allow only a maximum gain of $\epsilon$, where $0 \leq \epsilon < (1 - \frac{1}{n})$, to an agent from misreporting (note that we require $\epsilon < (1 - 
\frac{1}{n})$ since in any proportional mechanism the maximum an agent can gain by misreporting is equal to $(1 - \frac{1}{n})$). Unfortunately, we do not have an accompanying 
upper bound result for this and so the question of an upper bound remains open.

\begin{corollary}
For any $n \geq 2$ and $0 \leq \epsilon< \frac{1}{n}$, there is no deterministic, strategyproof, $\epsilon$-proportional, and non-wasteful 
mechanism for $n$ agents with piecewise constant valuations.
\end{corollary}

\begin{corollary}
For any $n \geq 2$ and $0 \leq \epsilon < \frac{1}{3n}$, there is no deterministic, $\epsilon$-strategyproof, proportional, and non-wasteful 
mechanism for $n$ agents with piecewise constant valuations.
\end{corollary}

\section{Removing non-wastefulness} ~\label{wm}
While non-wastefulness is certainly a desirable property (more so when, as mentioned in Section~\ref{nwm}, due to a result by \cite{sch96}, we know that Pareto-efficiency 
along with strategyproofness and fairness cannot be achieved), it seems like the negative results in the previous section were largely driven by this constraint as it severely 
restricts the kind of allocations that a mechanism can make. So, what if we remove this constraint? Do similar impossibilities still hold, or are there mechanisms that satisfy 
some of those properties? These are the questions we try to answer in this section.

\citet{chen10,chen13} provided a deterministic mechanism for piecewise uniform valuations and one of the main open questions they had posed was on generalizing their 
mechanism for piecewise constant valuations. Subsequently, \citet{aziz14} and \citet{branzei15} had posed the same question regarding the existence of mechanisms that are 
strategyproof and proportional for piecewise constant valuations. Here we address a special case of this question when the mechanism makes contiguous 
allocations and show that there exists no deterministic direct-revelation mechanism that is strategyproof and $\epsilon$-proportional, for any $0 \leq \epsilon < 
\frac{1}{n}$. Additionally, we also prove a stronger result (Proposition~\ref{prop2agents}) for the 
case of two agents. Informally, at a high-level, the proofs of both of these results follow the same theme as in the proof of Theorem~\ref{nwImp} where we construct valuation 
profiles with an aim of arriving at a contradiction. However, unlike in Theorem~\ref{nwImp}, we now no longer have the non-wasteful constraint, and so this requires some 
additional arguments which use the fact that the allocations are contiguous. 

\begin{theorem} \label{SP-prop-cont}
 For any $n \geq 2$ and $0 \leq \epsilon < \frac{1}{n}$, there is no deterministic mechanism for $n$ agents with piecewise constant valuations that makes contiguous 
allocations and is strategyproof and $\epsilon$-proportional.
\end{theorem}

Before we prove this, we prove the following stronger proposition which handles the two agent case. 

\begin{proposition} \label{prop2agents}
 For any $0 \leq \epsilon_1, \epsilon_2 < \frac{1}{2}$ such that $\epsilon_1 + \epsilon_2 < \frac{1}{2}$, there exists no deterministic mechanism for two hungry agents with 
piecewise constant valuations that makes contiguous allocations and is $\epsilon_1$-strategyproof and $\epsilon_2$-proportional.     
\end{proposition}

\begin{proof}
  Let us assume that there exists such a mechanism $\mathcal{M}$. Now, consider the preference profile $(u, u)$ of the agents, where both agents 1 and 2 have a valuation 
function $u$ that is uniform over $[0, 1]$, and let $\mathcal{M}$ make an allocation $(A_1, A_2)$, where $A_i$ is the cake allocated to agent $i$, in this profile. Since 
$\mathcal{M}$ outputs contiguous 
allocations, there is a point $c_1$ where the cut is made. As a result, we have a left piece ($[0, c_1]$) and a right piece ($[c_1, 1]$), and without loss of generality let 
us assume that agent 1 gets the left piece in the profile $P_1 = (u, u)$. Throughout, we refer to an allocation as being of LR-type (and correspondingly of being
RL-type) if agent 1 gets the left (right) piece (for some cut point $c$) and agent 2 gets the right (left) piece.  Since $\mathcal{M}$ is $\epsilon_2$-proportional and 
$\epsilon_2 < \frac{1}{2}$, $c_1 > 0$. Also, for the rest of the proof let us assume that $c_1 \geq \frac{1}{2}$ (we can make a similar construction as below if $c_1 < 
\frac{1}{2}$). 

Next, consider the profile $P_2 = (v, u)$, where $v$ is defined as follows.
\begin{equation*}
 \begin{aligned}[c]
v(x) = \begin{cases}       
       \frac{\frac{1}{2} + \epsilon_2}{\delta_1} & x \in [0, \delta_1] \\[1ex]
       \frac{\frac{1}{2} - \epsilon_1 - \epsilon_2 - \delta_2}{c_1 - \delta_1 - \delta_3} & x \in [\delta_1, c_1 - \delta_3] \\[1ex]
       \frac{\epsilon_1}{\delta_3} & x \in [c_1 - \delta_3, c_1] \\[1ex]
       \frac{\delta_2}{1 - c_1} & x \in [c_1, 1] 
   \end{cases}
\end{aligned}
\end{equation*} 
where $0 < \delta_2 < \frac{1}{2} - \epsilon_1 - \epsilon_2$, $0 < \delta_3 < c_1 - \epsilon_1$, and $0 < \delta_1 < 
\min\{c_1 - \epsilon_1 - \delta_3, \, \frac{1}{2} - \epsilon_2\}$ . 

Since $\mathcal{M}$ is $\epsilon_1$-strategyproof, it has to allocate agent 1 a value of at least $(1 - \epsilon_1 - \delta_2)$, for otherwise agent 1 can deviate to 
$P_1$ 
and as a result gain more than $\epsilon_1$. Now, since $\delta_1 < \frac{1}{2} - \epsilon_2$ and $\delta_2 < \frac{1}{2} - \epsilon_1 - \epsilon_2$, it can seen that an 
RL-type allocation is not possible. Hence $\mathcal{M}$ outputs an LR-type allocation with a cut point at, say, $c_2$, and since we know that agent 1 has to get a value of at 
least $1 - \epsilon_1 - \delta_2$, we have that $c_2 \geq c_1 - \delta_3$. This in turn implies that agent 2 gets at most $[c_1 - \delta_3, 1]$. 

Finally, consider $P_3 = (v, w)$, where $w$ is defined as follows. 
\begin{equation*}
 \begin{aligned}[c]
w(x) = \begin{cases}       
       \frac{\frac{1}{2} - \epsilon_2}{\delta_4} & x \in [0, \delta_4] \\[1ex]
       \frac{2\epsilon_2}{\delta_4 - \delta_5} & x \in [\delta_4, \delta_5] \\[1ex]
       \frac{\delta_4}{c_1 - \delta_3 - \delta_5} & x \in [\delta_5, c_1 - \delta_3] \\[1ex]
       \frac{\frac{1}{2} - \epsilon_2 - \delta_4}{1 - c_1 + \delta_3} & x \in [c_1 - \delta_3, 1]
   \end{cases}
\end{aligned}
\end{equation*} 
where $\delta_1 < \delta_4 < \delta_5 < c_1 - \epsilon_1 - \delta_3$. 

Now, here, since $\mathcal{M}$ is $\epsilon_2$-proportional and $\delta_4 > \delta_1$, it can be seen that an RL-type allocation is not possible. Therefore, 
$\mathcal{M}$ makes an LR-type allocation and since it is $\epsilon_2$-strategyproof and the whole cake needs to be allocated (since the agents are hungry and so every part 
of the cake is desired by them), agent 2 gets at least $[\delta_5, 1]$. But 
then, this implies that agent 2 in Profile 2 can deviate to Profile 3 and as a result gain a value of $c_1 - \delta_3 - \delta_5$ from 
it. And since $c_1 - \delta_3 - \delta_5 > \epsilon_1$, this contradicts our assumption that $\mathcal{M}$ is $\epsilon_1$-strategyproof.  
\end{proof}

As an aside, note that the $\epsilon_2 = 0$ case in Proposition~\ref{prop2agents} essentially says that if we insist on contiguous and proportional allocations, then one 
cannot provide any guarantee on strategyproofness (i.e., such mechanisms cannot be $\epsilon$-strategyproof for any feasible $\epsilon$). Later on, in 
Theorem~\ref{ep-SP-gen}, we show how one can circumvent this and obtain a $\frac{1}{4}$-strategyproof and proportional mechanism for two agents by giving each of them at most 
two contiguous pieces.   

Having proved the two agent case, next we look at the case when there are three or more agents. 

\begin{proof}[Proof of Theorem~\ref{SP-prop-cont}]
We start off by introducing the following observations and terminologies which we will use throughout this proof. 
\begin{itemize}
\item We say that a contiguous piece of cake $X = [a_1, b_1]$ is to the \textit{left} of another piece $Y = [a_2, b_2]$ if $a_1 < b_1 \leq a_2 < b_2$. 

\item Since we are interested in allocations $(A_1, \cdots, A_n)$ that are contiguous, there exists an underlying allocation order, say, $O$, for the agents in $[n]$. That is, 
we can obtain an order, say, $O = p_1, \cdots, p_n$, where $p_i \in [n]$ and agent $p_i$ gets a 
piece of cake that is to left of the piece allocated to agent $p_{i+1}$.

\item In all the valuation profiles defined below, every part of the cake is desired by at least one agent. And so, since we assume that a mechanism allocates all the pieces 
of cake that are valued at greater than zero by at least one of the agents, therefore in all the cases below we will have one agent who gets the piece $[c, 1]$ for some $c > 
0$. We refer to piece that contains the right endpoint of the cake (which is 1) as the \textit{right-end}. Therefore, when we say that an agent $i$ gets the right-end, this 
implies that agent $i$ is the final agent in the underlying allocation order and that she gets a piece of cake that has the right endpoint of the cake. Correspondingly, we 
say that an agent gets a \textit{left-end} if she gets a piece $[0, c]$ for some $c > 0$ (i.e., if she gets piece that has the left endpoint).  
\end{itemize}  

Equipped with the above, let us assume for the sake of contradiction that there exists a mechanism $\mathcal{M}$ that is deterministic, strategyproof, and 
$\epsilon$-proportional for $n \geq 3$ agents (the $n=2$ case has already been proved in Proposition~\ref{prop2agents}). Next, let us consider the profile $(u, u, r, \cdots, 
r)$, where $u$ is the uniform function on $[0, 1]$ and
\begin{equation*}
  \begin{aligned}[c]
r(x) = \begin{cases} \frac{n}{1-n\epsilon} & x \in \left[(\frac{1}{n} - \epsilon)^2, (\frac{1}{n} - \epsilon)^2 + \frac{1}{n} - \epsilon\right] \\ 0 & \text{otherwise}. 
\end{cases}
 \end{aligned}
\end{equation*}

For the profile under consideration, let $\mathcal{M}$ make an allocation $(A_1, \cdots, A_n)$ and let the underlying allocation order be $O$. First thing to notice is that 
the allocation order here can take only the following forms: $\{X, 2, 1\}$ or $\{X, 1, 2\}$, where $X$ is some permutation of $\{3, \cdots, n\}$. This is so because of the 
following reasons: first is that neither agent $1$ nor agent $2$ can be in between the agents in $X$ (i.e., the agents in $\{3, \cdots, n\}$ have to be consecutive (in some 
order) in $O$) as this would result in one of those agents not getting an $\epsilon$-proportional allocation (this is true because all these agents only value the interval 
$\left[(\frac{1}{n} - \epsilon)^2, (\frac{1}{n} - \epsilon)^2 + \frac{1}{n} - \epsilon\right]$ and so if agent 1 or 2 are in between, then since their allocation needs to be 
of 
length at least $(\frac{1}{n} - \epsilon)$, agents that appear to the right of them in $O$ will not receive an $\epsilon$-proportional allocation); the second reason is that 
none of the agents in $X$ can get the right-end as this would result in a cake of length strictly less than $(\frac{2}{n} - 2\epsilon)$ to allocate to agent 1 and agent 2; 
finally, agent 1 or 2 cannot get the left-end since this would again result in one of the agents in $X$ not getting an $\epsilon$-proportional allocation. Hence, the above 
arguments imply that either agent 1 or agent 2 receives the right-end in this profile. For the rest of the proof, we assume, without loss of generality, that agent 2 gets the 
right-end (we can make a similar construction if instead 
agent 1 is the one who gets it).  

Next, let us consider the profile $(v, u, r, \cdots, r)$, where
\begin{equation*}
  \begin{aligned}[c]
  v(x) = \begin{cases} \frac{1}{|A_1| } & x \in A_1 \\ 0 & \text{otherwise}. \end{cases}
 \end{aligned}
\end{equation*}

Since the mechanism is strategyproof, agent 1 has to get the allocation $A_1$ in this profile, because if otherwise, then she can deviate to profile 1. Additionally, we claim 
that agent 2 gets the right-end even in this profile. This is because of the following reasons: first, again, for the reasons stated previously, none of the agents in $X$ can 
get the right-end; second, agent 1 cannot get the right-end because this would entail that agent 1 gets a piece $A_1'$ such that $|A'_1| > |A_1|$ (since $A_1$ was not in the 
end in profile 1), which in turn will imply that agent 1 in profile 1 can deviate to profile 2. Hence, agent 1 receives $A_1 = [c_1, c_2]$, where $c_2 > (\frac{1}{n} - 
\epsilon)^2 + \frac{1}{n} - \epsilon$ (as $|A_1| \geq \frac{1}{n} - \epsilon$), and agent 2 receives an allocation $A_2' = [c_2, 1]$.  

Finally, consider the profile $(v, w, r, \cdots, r)$, where $0 < \delta < (\frac{1}{n} - \epsilon), b = \max\{c_2 - (\frac{1}{n} - \epsilon)(c_2 - c_1), (\frac{1}{n} - 
\epsilon)^2 + \frac{1}{n} - \epsilon\}$, and  
\begin{equation*}
  \begin{aligned}[c]
  w(x) = \begin{cases} 
  0 & x \in [0, b]\\[1ex]
  \frac{1- \frac{1}{n} + \epsilon + \delta}{c_2 - b } & x \in [b, c_2] \\[1ex]   
  \frac{\frac{1}{n} - \epsilon - \delta}{1 - c_2} & x \in [c_2, 1].
  \end{cases}
 \end{aligned}
\end{equation*}

Now, in the profile under consideration, again, it is easy to see that none of the agents in $X$ can receive the right-end, for this would result in agent 2 (with $w(x)$) 
getting a zero allocation. Also, agent 1 cannot get the right-end because we know that if agent 1 were to receive the right-end, then she would need to receive at least $[b, 
1]$ which in turn would again result in agent 2 getting a zero allocation. Therefore, agent 2 has to receive the right-end here and this in turn implies that he has to 
receive a piece $[c, 1]$, where $c < c_2$. However, this would result in agent 2 deviating from profile 2 to profile 3, thus contradicting the strategyproofness of 
$\mathcal{M}$. 
\end{proof}

Given Theorem~\ref{SP-prop-cont}, the first natural question that arises is if we can extend it to the case when the mechanism can make arbitrary allocations (i.e., when the 
mechanism can output non-contiguous allocations). Unfortunately, we do not have an answer to this question. Instead, below, we show a connection between direct-revelation 
mechanisms and mechanisms in the Robertson-Webb model for the case when agents have piecewise constant valuations. We believe that it will be useful in either of the 
cases---i.e., if the answer to the above question is either in the positive or if we want to prove that such a mechanism does not exist. In the case that such a mechanism 
exists, 
the connection basically shows that when given an upper bound on the maximum number of breakpoints in any agent's valuation function, one can construct a mechanism in the 
Robertson-Webb model that approximately achieves both of the properties. And in the case that one wants to prove that such a mechanism does not exist, one can use this 
connection to prove a non-existence (of mechanisms that are $\epsilon$-strategyproof and $\epsilon$-proportional for some $\epsilon > 0$) in the Robertson-Webb model and then 
map it back to show that there is no (finite) direct-revelation mechanism that is strategyproof and proportional. 

\begin{proposition} \label{approProp}
Let $P_1, P_2, P_3$ denote the property of strategyproofness, proportionality, and envy-freeness, respectively, and let $P \subseteq \{P_1, P_2, P_3\}$. Given $k$, an upper 
bound on the maximum number of breakpoints in any agent's valuation function, and any $n \geq 2$, if there exists a (finite) deterministic direct-revelation mechanism for $n$ 
agents with piecewise constant valuations that satisfies $P_i$, $\forall P_i \in P$, then, for any $\epsilon > 0$, there exists a mechanism in the Robertson-Webb model for $n$ 
agents with piecewise constant valuations that asks at most $n \left\lfloor \frac{2k}{\epsilon} \right\rfloor $ queries on each input and is $\epsilon$-$P_i$, $\forall P_i \in 
P$.
\end{proposition}

\begin{proof}
As can be seen from the statement of the proposition, the main objective here is to draw a connection between direct-revelation mechanisms and mechanisms in the 
Robertson-Webb model. Since in the latter model we only have query access to the valuation functions of the agents, the first step is to basically ``learn'' their functions to 
a sufficiently good approximation using only \textbf{eval} and \textbf{cut} queries. Subsequently, once we have that, we can essentially feed these functions to  the 
direct-revelation mechanism in order to create a mechanism in the Robertson-Webb model that achieves an approximate version of all the properties (among strategyproofness, 
proportionality, and envy-freeness) satisfied by the direct-revelation mechanism. 

For the first step, given that an agent has a piecewise constant valuation function and $k$ is an upper bound on the maximum number of breakpoints in their valuation 
function, the key observation is that when given an $\epsilon > 0$ we can $\epsilon$-approximate their valuation function $v(x)$ with another piecewise constant function 
$w(x)$ such that for any piece of cake $X$, $ \int_X v(x)\, dx - \frac{\epsilon}{2} \: \leq \: \int_X w(x)\, dx \: \leq \: \int_X v(x)\, dx + \frac{\epsilon}{2}$. We show 
that this can be done with $\left \lfloor \frac{2k}{\epsilon} \right \rfloor$ \textbf{cut} queries in the following claim.

\begin{claim}
 Given an $\epsilon > 0$, we can use $\left\lfloor \frac{2k}{\epsilon} \right\rfloor$ \textbf{cut} queries to $\epsilon$-approximate any piecewise constant function $v(x)$ 
that has at most $k$ breakpoints by another piecewise constant function $w(x)$ such that $W([0, 1]) = 1$ and for any piece of cake $X$,  
 \begin{equation} \label{eqn18}
 \int_X v(x)\, dx - \frac{\epsilon}{2} \: \leq \: \int_X w(x)\, dx \: \leq \: \int_X v(x)\, dx + \frac{\epsilon}{2}.
\end{equation}
\end{claim}

\begin{proof} \let\qed\relax

Algorithm~\ref{algo4} describes how to convert $v(x)$ to $w(x)$ by using $\left\lfloor \frac{2k}{\epsilon} \right\rfloor$ \textbf{cut} queries. First thing to notice is that 
$W([0, 1]) = 1$ since $W([0, 1]) = \int_{x_0}^{x_1} c_1 + \cdots + \int_{x_{N-1}}^{x_N} c_N + \left( 1 - \frac{\epsilon}{2k} \cdot N\right) = 1$. So now, the only thing 
that's left to prove is to show that for any piece of cake $X$, $ \int_X v(x)\, dx - \frac{\epsilon}{2} \: \leq \: \int_X w(x)\, dx \: \leq \: \int_X v(x)\, dx + 
\frac{\epsilon}{2}$.

\begin{algorithm}[tb]
\begin{center}
\noindent\fbox{%
\begin{varwidth}{\dimexpr\linewidth-2\fboxsep-2\fboxrule\relax}
\begin{algorithmic}[1]
  \Input  An agent $i$'s piecewise constant valuation function $v_i(x)$, $\epsilon > 0$, and $k$ 
  \Output A piecewise constant function $w_i(x)$ that $\epsilon$-approximates $v_i(x)$
  \State $N \leftarrow \left \lfloor \frac{2k}{\epsilon} \right \rfloor$
  \State $x_0 \leftarrow 0, \: x_{N+1} \leftarrow 1$  
  \For{$j \in \{1, \cdots, N\}}$
    \State $x_j \leftarrow$ \textbf{cut}$\left(i, x_{j-1}, \frac{\epsilon}{2k}\right)$
    \State $c_j \leftarrow \frac{\epsilon}{2k \cdot (x_{j} - x_{j-1})}$  
    \State $w_i(x) = c_j, \:  x \in [x_{j-1}, x_j]$                       
  \EndFor
  \If {$X_N \neq 1$}
    \State $w_i(x) = \frac{\left( 1 - \frac{\epsilon}{2k} \cdot N\right)}{x_{N+1} - x_N}, \:  x \in [x_{N}, x_{N+1}]$ 
 \EndIf   
  \State return $w_i(x)$   
\end{algorithmic}
\end{varwidth}
}
\end{center}
\caption{$\epsilon$-approximation of a piecewise constant function using $\left\lfloor \frac{2k}{\epsilon} \right\rfloor$ \textbf{cut} queries.}
\label{algo4}
\end{algorithm}

To prove this, consider the case when $X$ is a collection of finitely many disjoint intervals (i.e., it is any arbitrary piece of cake). If we denote the breakpoints of $v$ 
by 
$\{b_1, \cdots, b_k\}$ (since there are only a maximum of $k$ breakpoints), where $0 = b_0  < b_1 < \cdots < b_k < b_{k+1} = 1$, and if we let $C_i = X \cap [x_{i-1}, x_i]$, 
then we know that $\int_{X} w(x)\, dx = \sum_{i=1}^{N+1} \int_{C_i} w(x)\, dx$. Additionally, also note that if $[x_{i-1}, x_{i}]$ does not contain any breakpoint $b_j$ of 
the original function $v(x)$, then $\int_{C_i} w(x)\, dx = \int_{C_i} v(x)\, dx$ since we know that $v(x)$ is constant over this interval (as there are no breakpoints) and so 
$v(x) = \frac{\epsilon}{2k \cdot (x_i - x_{i-1})} = w(x) , \forall x \in [x_{i-1}, x_{i}]$. This, along with the fact that $v(x)$ has at most $k$ breakpoints, implies that 
there are at most $k$ $C_i$'s for which $\int_{C_i} w(x)\, dx \neq \int_{C_i} v(x)\, dx$. Now, if we assume without loss of generality that these are $C_1 \cdots, C_k$, then 
using the fact that $\int_{C_i} w(x)\, dx \leq \ \frac{\epsilon}{2k}$ and $\int_{C_i} v(x)\, dx \leq  \frac{\epsilon}{2k}$ we have that for every $i \in \{1, \cdots, k\}$,
\begin{multline}\label{eqn22}
 \int_{C_i} w(x)\, dx \, - \, \frac{\epsilon}{2k} \, \leq \, 0  \, \leq \, \int_{C_i} v(x)\, dx \, \leq \, \frac{\epsilon}{2k} \,\leq \,\int_{C_i} w(x)\, dx \, + \, 
\frac{\epsilon}{2k}  \\
 \implies - \frac{\epsilon}{2k} \, \leq \, \int_{C_i} w(x)\, dx \, - \,  \int_{C_i} v(x)\, dx \,\leq \, \frac{\epsilon}{2k}.
\end{multline}
Therefore, since
\begin{align*}
 \int_{X} w(x)\, dx \, - \,  \int_{X} v(x)\, dx \, &= \, \sum_{i=1}^{N+1} \int_{C_i} w(x)\, dx \, - \,  \sum_{i=1}^{N+1} \int_{C_i} v(x)\, dx\\
 &= \, \sum_{i = 1}^{k} \int_{C_i} w(x)\, dx \, -  \, \sum_{i=1}^{k} \int_{C_i} v(x)\, dx, 
 \end{align*}
 we can now use equation~\ref{eqn22} to see that,
\begin{align*}
 &-\frac{\epsilon}{2} \, \leq \, \int_{X} w(x)\, dx \, - \, \int_{X} v(x)\, dx \, \leq \, \frac{\epsilon}{2}, 
\end{align*}
thus proving our claim that $W$ is an $\epsilon$-approximation of $V$. \hfill $\blacksquare$
\end{proof}

Now that we know that any utility function of an agent can be $\epsilon$-approximated using $\left \lfloor \frac{2k}{\epsilon} \right \rfloor$ \textbf{cut} queries, we can 
build a mechanism, $\mathcal{M}^{RW}$, in the Robertson-Webb model by just feeding these approximated valuation functions into the direct revelation mechanism $\mathcal{M}$. 
More formally, if $(v_1, \cdots, v_n)$ are the original valuation functions of the agents  and $(w_1, \cdots, w_n)$ are the corresponding $\epsilon$-approximated valuation 
functions, and $\mathcal{M}_i(v_1, \cdots, v_n)$ denotes the allocation made by $\mathcal{M}$ to agent $i$, then $\mathcal{M}^{RW}$ is such that,
\begin{equation*}
 \forall i \in \{1, \cdots, n\}, \quad \mathcal{M}_i^{RW} (v_1, \cdots, v_n) \: = \: \mathcal{M}_i (w_1, \cdots, w_n).  
\end{equation*}

It is clear that $\mathcal{M}^{RW}$ makes at most $n \left\lfloor \frac{2k}{\epsilon} \right\rfloor $ queries on each input and  so the only part that remains to be proved is 
to show that $\mathcal{M}^{RW}$ is $\epsilon$-$P_i$ for every property $P_i \in P$ satisfied by $\mathcal{M}$. To prove this let us consider each of the properties separately.

{{\textbf{i) When $\mathcal{M}$ satisfies $P_1$ (strategyproofness).}}} Let us assume without loss of generality that agent $i$ is the one trying to manipulate. 
Since the underlying direct-revelation mechanism is strategyproof, the only way for agent $i$ with a valuation function $v_i$ to manipulate $\mathcal{M}^{RW}$ is to pretend 
as 
if its function is some $v_i'$ and answer the queries accordingly. This in turn will result in $w_i'$ being the $\epsilon$-approximated function instead of $w_i$ if it 
reported truthfully. However, since $\mathcal{M}$ is strategyproof, we know that $\forall \: w_i', \, w_{-i}, \: W_i(\mathcal{M}_i(w_i, w_{-i})) \geq W_i(\mathcal{M}_i(w'_i, 
w_{-i}))$. This in turn implies that,
 \begin{multline*}
  V_i(\mathcal{M}_i(w_i, w_{-i})) \: + \: \frac{\epsilon}{2} \: \geq  \: W_i(\mathcal{M}_i(w_i, w_{-i})) \geq \\ W_i(\mathcal{M}_i(w'_i, w_{-i})) \: \geq \: 
V_i(\mathcal{M}_i(w'_i, w_{-i})) \: - \: \frac{\epsilon}{2}  \quad \text{\small (using equation~\ref{eqn18})} \\
\implies V_i(\mathcal{M}_i^{RW}(v_i, v_{-i})) \: + \: \frac{\epsilon}{2}  \: \geq \: W_i(\mathcal{M}_i^{RW}(v_i, v_{-i})) \: \geq \: \\ W_i(\mathcal{M}_i^{RW}(v'_i, v_{-i})) 
\: \geq \: V_i(\mathcal{M}_i^{RW}(v'_i, v_{-i})) \: - \: \frac{\epsilon}{2}. 
 \end{multline*}
And hence $ \forall \: v_i', \, v_{-i}$ we have that $V_i(\mathcal{M}_i^{RW}(v'_i,v_{-i})) \: - \: V_i(\mathcal{M}_i^{RW}(v_i, v_{-i})) \leq \epsilon$, or in other words 
that $\mathcal{M}^{RW}$ is $\epsilon$-strategyproof.  

{{\textbf{ii) When $\mathcal{M}$ satisfies $P_2$ (proportionality).}}} To prove that $\mathcal{M}^{RW}$ is ${\epsilon}$-proportional note that $\forall 
i$, $W_i(\mathcal{M}_i(w_i, w_{-i})) \geq \frac{1}{n}$ since $\mathcal{M}$ is proportional. This implies that, using equation~\ref{eqn18}, we have, 
\begin{multline*}
 V_i(\mathcal{M}_i(w_i, w_{-i})) \: + \: \frac{\epsilon}{2} \: \geq \: W_i(\mathcal{M}_i(w_i, w_{-i})) \: \geq \: \frac{1}{n}\\
 \implies V_i(\mathcal{M}_i^{RW}(v_i, v_{-i})) \: = \: V_i(\mathcal{M}_i(w_i, w_{-i})) \: \geq \: \frac{1}{n} \: - \frac{\epsilon}{2} \: > \: \frac{1}{n} \: - \: 
{\epsilon}.
\end{multline*}

{{\textbf{iii) When $\mathcal{M}$ satisfies $P_3$ (envy-freeness).}}} $\mathcal{M}^{RW}$ is ${\epsilon}$-envy-free since $\forall j \in [n]$, using 
equation~\ref{eqn18} and the fact that $\mathcal{M}$ is envy-free, we have,  
\begin{multline*}
 V_i(\mathcal{M}_i(w_i, w_{-i})) \: + \: \frac{\epsilon}{2} \: \geq \: W_i(\mathcal{M}_i(w_i, w_{-i})) \: \geq \\ \: W_i(\mathcal{M}_j(w_i, w_{-i})) \: \geq \: 
 V_i(\mathcal{M}_j(w_i, w_{-i})) \: - \: \frac{\epsilon}{2} \\
 \implies V_i(\mathcal{M}_i^{RW}(v_i, v_{-i})) \: \geq \: V_i(\mathcal{M}_j^{RW}(v_i, v_{-i})) \: - \: \epsilon.
\end{multline*}
This in turn proves our proposition. 
\end{proof}

Although, as mentioned above, it remains open as to whether a strategyproof and proportional mechanism exists, another goal one could still consider is on 
finding better proportional mechanisms (in terms of $\epsilon$-strategyproofness). Below, we begin the pursuit of this goal and as our final question in this paper ask if 
there are positive results at least regarding mechanisms that are proportional and $\epsilon$-strategyproof for some $\epsilon < (1 - \frac{1}{n})$. We answer it with a 
``Yes'' and in particular we first prove---by showing a bound on the maximum gain an agent can get by misreporting---that the well-known Even-Paz algorithm \citep{even84} 
satisfies the above criterion for the case when there are at least four agents. And following this, we present a minor modification to the same that has better incentive 
compatible properties than the original Even-Paz algorithm for the cases when there are two or three agents. For the sake of completeness, the Even-Paz algorithm is presented 
in Appendix~\ref{appendixB}.

\begin{proposition} \label{EPbounds}
 For $n \geq 2$ agents, the Even-Paz algorithm is deterministic, proportional, and $\epsilon$-strategyproof, where 
  \begin{equation*}
 \begin{aligned}[c]
\epsilon = \begin{cases} 
      \frac{1}{2} & \text{when $n = 2$ or $n = 4$} \\
      \frac{2}{3} & \text{when $n = 3$ or $n = 5$} \\
      1 - \frac{2}{n} & \text{when $n \geq 6$.}
   \end{cases}
\end{aligned}
\end{equation*}
\end{proposition}

\begin{proof}
 
It is well known that the Even-Paz algorithm is deterministic and proportional. So below we only derive the bounds on $\epsilon$. Also, for the sake of simplicity, we keep 
the proof largely informal.

For the $n = 2$ and $n = 3$ cases, it is easy to construct examples where the gain from lying is as much as $(1 - \frac{1}{n})$---i.e., it is $\frac{1}{2}$ 
and $\frac{2}{3}$, respectively. So let us consider the case when there are $n \geq 4$ agents. Notice that in every instance of the Even-Paz algorithm (see 
Appendix~\ref{appendixB}), the procedure EP is run multiple times and during each run the procedure asks an agent $i$ for a cutpoint $c_i$ such that $V_i(a, c_i) = 
\frac{\lfloor \frac{k}{2} \rfloor}{k} \cdot V_i(a, b)$, where $[a, b]$ is the cake to be proportionally allocated among the $k = |S|$ agents in $S$. Throughout, we refer to 
each such run of the procedure EP as a ``step'' of the Even-Paz algorithm; therefore, the first step refers to the first run of the procedure EP where an agent $i$ is asked 
for a cutpoint $c_i$ such that $V_i(0, c_i) = \frac{\lfloor \frac{n}{2} \rfloor}{n}$.  

So now, let us consider the first step of the algorithm. Without loss of generality, we can assume that for $i, j \in \{1, \cdots, n\}$, if $i < j$, then $c_i \leq c_j$ 
(since the agents can always be relabeled to satisfy this). Additionally, let us call the piece $[0, c_{\lfloor \frac{n}{2} \rfloor}]$ as the `left piece', the piece 
$[c_{\lfloor \frac{n}{2}\rfloor + 1}, 1]$ as the `right piece', the piece $[c_{\lfloor \frac{n}{2} \rfloor}, c_{\lfloor \frac{n}{2}\rfloor + 1}]$ as the `middle piece', and 
also denote by $G(n)$ the maximum gain an agent can get by misreporting when there are $n$ agents. 

Next, notice that in the Even-Paz algorithm every agent either recurses on the left piece or on the piece $[c_{\lfloor \frac{n}{2} \rfloor}, 1]$ (which 
is the $\{$middle $\cup$ right$\}$ piece). So let us consider the following two cases separately: i) an agent $i$ who if reporting truthfully in the first step would recurse 
on the left piece (i.e., whose true report $c_i^t$ is such that $c^{t}_i \leq c_{\lfloor \frac{n}{2}\rfloor}$) and ii) an agent $j$ who if reporting truthfully in the first 
step would recurse on the piece $[c_{\lfloor \frac{n}{2} \rfloor}, 1]$ (i.e., whose true report $c_j^{t}$ is such that $c^{t}_j > c_{\lfloor 
\frac{n}{2}\rfloor}$). We will denote the gain for an agent in case i) and ii) to be $G^L(n)$ and $G^R(n)$, respectively. Also note that since every agent belongs to either 
case i) or case ii), the maximum any agent can gain by lying is $\max\{G^L(n), G^R(n)\}$.  

 {{\textbf{Case i): an agent $i$ such that $c^{t}_i \leq c_{\lfloor \frac{n}{2}\rfloor}$.}}} First, if $c^{t}_i \leq c_{\lfloor \frac{n}{2}\rfloor}$ is the true 
report of $i$ in the first step, then, when playing strategically, the maximum $i$ can gain by reporting a $c_i > c_{\lfloor \frac{n}{2}\rfloor + 1}$ is at most 
$(\frac{\lceil \frac{n}{2} \rceil}{n} - \frac{1}{n})$. This is so because she is guaranteed to get $\frac{1}{n}$ when playing truthfully (since the mechanism is proportional) 
and she gets at most $V_i(c_{\lfloor \frac{n}{2}\rfloor + 1}, 1) \leq \frac{\lceil \frac{n}{2} \rceil}{n}$ when misreporting (since the other players have not 
changed their reports). Therefore, given the fact that upper bound we want to prove is greater than $(\frac{\lceil \frac{n}{2} \rceil}{n} - \frac{1}{n})$, we need to only 
consider the case when $i$ misreports a $c_i \leq c_{\lfloor \frac{n}{2}\rfloor + 1}$. 

So now, let us consider the case when $i$ misreports a $c_i \leq c_{\lfloor \frac{n}{2}\rfloor + 1}$. Note that even in this case, in order to maximize the gain, $i$ needs to 
report $c_i$ in such a way that it expands the cake she's recursing on (i.e., report a $c_i$ such that $c_{\lfloor \frac{n}{2}\rfloor} < c_i \leq c_{\lfloor 
\frac{n}{2}\rfloor + 1}$). This is so because, if $i$ does not expand the left piece (which is the piece it is recurses on if reporting truthfully), then the maximum gain 
$i$ can gain by misreporting in the algorithm is $\left(V_i(0, c_{\lfloor \frac{n}{2}\rfloor})  -  \frac{V_i(0, c_{\lfloor \frac{n}{2}\rfloor})} {\lfloor 
\frac{n}{2}\rfloor}\right)$---the first term because that's the maximum cake $i$ has access to from the second level onwards and the second term because $i$ is guaranteed to 
get at least a $\frac {1}{\lfloor \frac{n}{2} \rfloor}$ fraction of $V_i(0, c_{\lfloor \frac{n}{2}\rfloor})$ if he reports truthfully.  And since $V_i(0, c_{\lfloor 
\frac{n}{2}\rfloor}) \leq 1$, this implies that the maximum gain $G^L(n) \leq 1 - \frac{1}{\lfloor \frac{n}{2}\rfloor} \leq 1 - \frac{2}{n}$. 

Therefore, we need to only consider the case when $i$ misreports and as a result expands the piece she recurses on (i.e., she reports a $c_i$ such that $c_{\lfloor 
\frac{n}{2}\rfloor} < c_i \leq c_{\lfloor \frac{n}{2}\rfloor + 1}$). Now, in this case, it can be seen that if she has to gain more than in the case when there is no 
expansion, then she has to get the middle piece (or at least some part of it; for the rest of the proof although we talk about the middle piece as a whole, all the 
arguments will be true even if $i$ is only getting a part of it). However, note that the algorithm always produces a contiguous allocation (i.e., all the agents are 
allocated a single continuous interval of the cake), and so the only way an agent can get the middle piece is if the piece she recurses on at every level contains this middle 
piece. 

Given this, let us consider the penultimate level where there are two agents ($i$ and some other agent, say, $k$). Here $i$ will get the middle piece only if her report is to 
the right of the other agent. Now, if the value of this piece (i.e., the one that is to the right of $k$'s report in the penultimate level) is greater than $(1 - 
\frac{1}{n})$ according to $i$'s true valuation, then it would have meant that the 
only part of the cake $i$ received extra as a result of lying is the middle piece. This is so because, every other agent is playing the same way irrespective of whether $i$ 
is lying or otherwise. Therefore, when $i$ is lying and she expands the piece at the first level, the key observation to make is that none of the reports of the other players 
``decrease''---i.e., in every level, whenever the agents are asked for their reports during the course of the algorithm, the reports that they specify are no less than what 
they would have specified in the case where $i$ was playing truthfully (this is so because, when the cake is expanded, the value for it for every agent increases or 
remains the same and so when the agents are asked for their reports, their reports can only be greater than the case when the cake wasn't expanded).  This in turn implies 
that agent $i$ when playing truthfully would have got everything except the middle piece (which was introduced as a result of lying) because if not then she would have ended 
up receiving less than $\frac{1}{n}$ which is impossible. Hence, following this line of reasoning, only the middle piece is extra and we know that the value of the middle 
piece is at most $\frac{\lceil \frac{n}{2} \rceil}{n} \leq 1 - \frac{2}{n}$. So now, the case that remains is when the value of this piece according to $i$'s true 
valuation is less than or equal to $(1 - \frac{1}{n})$. And in this case we know that the gain from lying is at most $(1 - \frac{2}{n})$ since she's guaranteed to get 
$\frac{1}{n}$ when playing truthfully. 

Therefore, in all the cases above, an agent $i$ such that $c^{t}_i \leq c_{\lfloor \frac{n}{2}\rfloor}$, cannot gain more than $(1 - \frac{2}{n})$. And hence $G^L(n) \leq (1 
- \frac{2}{n})$.

{\textbf{Case ii): an agent $j$ such that $c^{t}_j > c_{\lfloor \frac{n}{2}\rfloor}$.}} This case is much easier to argue. Note that if agent $j$ 
reports a $c_j < c_{\lfloor \frac{n}{2}\rfloor}$, then the gain is at most $(\frac{\lfloor \frac{n}{2} \rfloor}{n} - \frac{1}{n})$ and since the bound we want to prove is 
bigger we can ignore this. So now, if we ignore the case above, then it is easy to see that $j$ has no real control over the piece he's recursing on in the first 
step---unlike how it was in case i). Therefore, the worst case is when he has a value 1 for the piece he's recursing on. Since there are only ${\lceil \frac{n}{2} \rceil}$ 
players recursing along with her in this piece, we can treat this as a new instance with ${\lceil \frac{n}{2} \rceil}$ players. And so the maximum gain for $j$ in this case is 
at most $G({\lceil \frac{n}{2} \rceil})$---because if it is anything more then it would contradict the fact that $G({\lceil \frac{n}{2} \rceil})$ is the maximum an agent can 
gain in a ${\lceil \frac{n}{2} \rceil}$-player instance. 

Therefore, putting together cases i) and ii) we have, $ G(n) = \max\{G^L(n), G^R(n)\} = \max\{1 - \frac{2}{n},\, G\left({\left\lceil
\frac{n}{2}\right\rceil}\right)\}$. Now, we know that $G(2) = \frac{1}{2}, G(3) = \frac{2}{3}$, and so we get $G(4) = \frac{1}{2}$ and $G(5) = \frac{2}{3}$. For $n \geq 
6$, it is easy to see that $G(n) = 1 - \frac{2}{n}$ works. Hence, Even-Paz algorithm is deterministic, proportional and $\epsilon$-strategyproof, where 
 \begin{equation*}
 \begin{aligned}[c]
\epsilon = \begin{cases} 
      \frac{1}{2} & \text{when $n = 2$ or $n = 4$} \\
      \frac{2}{3} & \text{when $n = 3$ or $n = 5$} \\
      1 - \frac{2}{n} & \text{when $n \geq 6$.}
   \end{cases}
\end{aligned}
\end{equation*}
\end{proof}

Next, we show that the bounds on $\epsilon$ can be improved for the cases when there are two or three agents. In particular, we show that Algorithm~\ref{algo3} is 
proportional and $\epsilon$-strategyproof, where $\epsilon = (1- \frac{3}{2n})$ or $ \epsilon = (1-\frac{3}{2n} + \frac{1}{2n^2})$ depending on whether $n$ is even or 
odd, respectively. And so, for the cases when there are two or three agents, the gain is $\frac{1}{4}$ or $\frac{5}{9}$, respectively, as opposed to $\frac{1}{2}$ or 
$\frac{2}{3}$ in the Even-Paz algorithm. The main difference between the original Even-Paz algorithm (see Appendix~\ref{appendixB}) and the modified 
version presented here (Algorithm~\ref{algo3}) is in how at each level the piece $\left[d_{{\left\lfloor \frac{k}{2} \right\rfloor }}, d_{{\left\lfloor \frac{k}{2} 
\right\rfloor} + 1}\right]$ is handled. In the original Even-Paz algorithm (see procedure EP in Algorithm~\ref{algoEP} in Appendix~\ref{appendixB}), this piece is part of the 
piece that the agents in $S_R$ recurse on, while in the modifiedEP procedure above, at each level, this piece is proportionally split between the $k$ agents by again using the 
original Even-Paz algorithm.

\begin{algorithm}[t]
\begin{center}
\noindent\fbox{%
\begin{varwidth}{\dimexpr\linewidth-3\fboxsep-3\fboxrule\relax}
\begin{algorithmic}[1]
  \Procedurename modifiedEP$\left([a, b], S\right)$
   \Input $\left([a, b], S\right)$, where $[a, b] \subseteq [0, 1]$ is the cake to be proportionally allocated among $|S|$ agents 
  \State $k \leftarrow |S|$
  \If {$k == 1$} 
    \State \textbf{return} $([a, b])$
   \EndIf

  \For{each $i \in S$ }
    \State $c_i$ = \textbf{cut}$\left(i, a, \frac{\left\lfloor \frac{k}{2} \right\rfloor}{k} \cdot V_i(a, b)\right)$
  \EndFor
  
  \State Let $\{d_1, \cdots, d_{k} \}$ be the sorted order of the $c_i$'s (cutpoints), $S_L \leftarrow \left\{i \in S \mid c_i \leq d_{\left\lfloor \frac{k}{2} 
\right\rfloor}\right\}$, and  $S_R \leftarrow \left\{i \in S \mid c_i \geq d_{\left\lfloor \frac{k}{2} \right\rfloor + 1} \right\}$.

  \State Proportionally allocate the cake $\left[d_{{\left\lfloor \frac{k}{2} \right\rfloor }}, d_{{\left\lfloor \frac{k}{2} \right\rfloor} + 1}\right]$ 
among the $k$ agents using the Even-Paz algorithm. 

 \State Recursively call modifiedEP on $\left[a, d_{\left\lfloor \frac{k}{2} \right\rfloor}\right]$ and $\left[d_{\left\lfloor \frac{k}{2} \right\rfloor + 1}, 
b\right]$ with $S_L$ and $S_R$, respectively. Return the allocation formed by these recursive calls along with the proportional allocation of the middle piece (i.e., the 
piece $\left[d_{{\left\lfloor \frac{k}{2} \right\rfloor }}, d_{{\left\lfloor \frac{k}{2} \right\rfloor} + 1}\right]$).
\Statex
\Main 
     \State output the allocation returned by modifiedEP$\left([0, 1], S = \{1, \cdots, n\}\right)$
 \end{algorithmic}  
\end{varwidth}
}
\end{center}
\caption{Modified Even-Paz algorithm. }
\label{algo3}
\end{algorithm}

\begin{theorem} \label{ep-SP-gen}
 Algorithm~\ref{algo3} is deterministic, proportional, and $\epsilon$-strategyproof, where  
 \begin{equation*}
 \begin{aligned}[c]
\epsilon = \begin{cases} 
      1 - \frac{3}{2n} & \text{when $n$ is even} \\
      1 - \frac{3}{2n} + \frac{1}{2n^2} & \text{when $n$ is odd.} 
   \end{cases}
\end{aligned}
\end{equation*}
\end{theorem}

\begin{proof}
It is easy to see that Algorithm~\ref{algo3} is deterministic and proportional. To prove that it is $\epsilon$-strategyproof, we proceed like in the proof of 
Proposition~\ref{EPbounds} by considering the first step of the algorithm where each agent $i$ reports a point $c_i$ such that ${V_i(0, c_i)} = \frac{\lfloor \frac{n}{2} 
\rfloor}{n}$. Without loss of generality, we can assume that for $i, j \in \{1, \cdots, n\}$, if $i < j$, then $c_i \leq c_j$ 
(since 
the agents can always be relabeled to satisfy this). Also, it can be observed that all the agents ($1, \cdots, {\lfloor \frac{n}{2} \rfloor}$) recursing on the piece $[0, 
c_{\lfloor \frac{n}{2} \rfloor}]$ (henceforth also referred to as the `left piece') are symmetric---as in, after the first step all of them are treated in the same way. 
Similarly, all the agents (${\lfloor \frac{n}{2} \rfloor + 1}, \cdots, {n}$) recursing on the piece $[c_{\lfloor \frac{n}{2}\rfloor + 1}, 1]$ (henceforth also 
referred to as the `right piece') are also symmetric. Hence, for any agent $i$, and for their true report $c^t_i$ in step 1, we need to only consider two cases---i) $c^t_i 
\leq c_{\lfloor \frac{n}{2}\rfloor}$ (this corresponds to the case when $i$ will recurse on the left piece when playing  truthfully) and ii) $c^t_i \geq c_{\lfloor 
\frac{n}{2}\rfloor + 1}$ (this corresponds to the case when $i$ will recurse on the right piece when playing truthfully)---and then argue that the maximum $i$ can gain by 
reporting untruthfully is $(1- \frac{3}{2n})$ or $(1- \frac{3}{2n} + \frac{1}{2n^2})$ depending on whether $n$ is even or odd, respectively. 

 \textbf{Case i): $c^t_i \leq c_{\lfloor \frac{n}{2}\rfloor}$.} First, if $c^t_i \leq c_{\lfloor \frac{n}{2}\rfloor}$ is the true report of $i$, then, when playing 
strategically, if $i$ reports a $c_i > c_{\lfloor \frac{n}{2}\rfloor + 1}$, the maximum he can gain by doing so is at most $(\frac{\lceil \frac{n}{2} \rceil}{n} - 
\frac{1}{n})$. This is so because, he is guaranteed to get $\frac{1}{n}$ when playing truthfully (since the mechanism is proportional) and he gets at most $V_i(c_{\lfloor 
\frac{n}{2}\rfloor + 1}, 1) \leq \frac{\lceil \frac{n}{2} \rceil}{n}$ when playing untruthfully as the other players have not changed their reports. Therefore, given the fact 
that upper bound we want to prove is greater than $(\frac{\lceil \frac{n}{2} \rceil}{n} - \frac{1}{n})$, we need to only consider the case when $i$ misreports a $c_i \leq 
c_{\lfloor \frac{n}{2}\rfloor + 1}$. 

In this case, let $L$ be the maximum value he can get by lying. Since $i$ does not get any part of $[c_{\lfloor \frac{n}{2}\rfloor + 1}, 1]$ if he reports a $c_i \leq 
c_{\lfloor \frac{n}{2}\rfloor + 1}$, to get $L$ it has to be the case that $V_i(0, c_{\lfloor\frac{n}{2}\rfloor}) + V_i(c_{\lfloor\frac{n}{2}\rfloor}, 
c_{\lfloor\frac{n}{2}\rfloor 
+ 1}) \geq L$. On other hand, if $i$ plays truthfully, then he gets T, where 
\begin{align*}
 T &\geq \frac{1}{\lfloor \frac{n}{2} \rfloor} V_i(0, c_{\lfloor\frac{n}{2}\rfloor}) + \frac{1}{n} V_i(c_{\lfloor\frac{n}{2}\rfloor}, c_{\lfloor\frac{n}{2}\rfloor + 1}) \\ 
  &\geq \frac{1}{n} \left(V_i(0, c_{\lfloor\frac{n}{2}\rfloor}) + V_i(c_{\lfloor\frac{n}{2}\rfloor}, c_{\lfloor\frac{n}{2}\rfloor + 1})\right) + \left(\frac{1}{\lfloor 
\frac{n}{2} 
\rfloor} - 
\frac{1}{n}\right) V_i(0, c_{\lfloor\frac{n}{2}\rfloor}) \\
 &\geq \frac{1}{n} L + \left(\frac{1}{\lfloor \frac{n}{2} \rfloor} - \frac{1}{n}\right) \frac{\lfloor \frac{n}{2} \rfloor}{n}.
\end{align*}
 
Therefore, the maximum gain as a result of reporting untruthfully is,
\begin{align}
 G &\leq L - T \nonumber \\
 &\leq L\left(1- \frac{1}{n}\right) - \left(\frac{1}{\lfloor \frac{n}{2} \rfloor} - \frac{1}{n}\right) \frac{\lfloor \frac{n}{2} \rfloor}{n} \nonumber \\
 &\leq \left(1- \frac{1}{n}\right) - \left(\frac{1}{\lfloor \frac{n}{2} \rfloor} - \frac{1}{n}\right) \frac{\lfloor \frac{n}{2} \rfloor}{n}  &&\text{(since $L \leq 1$)} 
\nonumber \\ 
 &\begin{aligned}[c] \label{eq15}
 \leq  \begin{cases} 
      1- \frac{3}{2n} & \text{when $n$ is even} \\
      1- \frac{3}{2n} - \frac{1}{2n^2} & \text{when $n$ is odd}.
   \end{cases}
\end{aligned} 
\end{align}

\textbf{Case ii): $c^t_i \geq c_{\lfloor \frac{n}{2}\rfloor + 1}$.}  Like in case i), if $c^t_i \geq c_{\lfloor \frac{n}{2}\rfloor + 1}$ is the true report of $i$, 
then, 
when playing strategically, if $i$ reports a $c_i< c_{\lfloor \frac{n}{2}\rfloor}$, the maximum he can gain by doing so is at most $\left(\frac{\lfloor \frac{n}{2} 
\rfloor}{n} 
- \frac{1}{n}\right)$, which is less than the upper bound we want to prove. Hence, we need to only consider the case when $i$ misreports a $c_i \geq c_{\lfloor 
\frac{n}{2}\rfloor}$. 

In this case, let $L$ be the maximum value he can get by lying. Since $i$ does not get any part of $[0, c_{\lfloor \frac{n}{2}\rfloor}]$ if he reports a $c_i \geq c_{\lfloor 
\frac{n}{2}\rfloor}$, to get $L$ it has to be the case that $V_i(c_{\lfloor\frac{n}{2}\rfloor}, c_{\lfloor\frac{n}{2}\rfloor + 1} ) + V_i(c_{\lfloor\frac{n}{2}\rfloor + 1}, 
1) 
\geq L$. On other hand, if 
$i$ plays truthfully, then he gets T, where
\begin{align*}
 T &\geq \frac{1}{\lceil \frac{n}{2} \rceil} V_i(c_{\lfloor\frac{n}{2}\rfloor + 1}, 1) + \frac{1}{n} V_i(c_{\lfloor\frac{n}{2}\rfloor}, c_{\lfloor\frac{n}{2}\rfloor + 1}) \\ 
  &\geq \frac{1}{n} \left(V_i(c_{\lfloor\frac{n}{2}\rfloor + 1}, 1) + V_i(c_{\lfloor\frac{n}{2}\rfloor}, c_{\lfloor\frac{n}{2}\rfloor + 1})\right) + \left(\frac{1}{\lceil 
\frac{n}{2} \rceil} - 
\frac{1}{n}\right) V_i(c_{\lfloor\frac{n}{2}\rfloor + 1}, 1) \\
 &\geq \frac{1}{n} L + \left(\frac{1}{\lceil \frac{n}{2} \rceil} - \frac{1}{n}\right) \frac{\lceil \frac{n}{2} \rceil}{n}.
\end{align*}

Therefore, the maximum gain as a result of reporting untruthfully is,
\begin{align}
 G &\leq L - T \nonumber \\
 &\leq L\left(1- \frac{1}{n}\right) - \left(\frac{1}{\lceil \frac{n}{2} \rceil} - \frac{1}{n}\right) \frac{\lceil \frac{n}{2} \rceil}{n} \nonumber \\
 &\leq \left(1- \frac{1}{n}\right)  - \left(\frac{1}{\lceil \frac{n}{2} \rceil} - \frac{1}{n}\right) \frac{\lceil \frac{n}{2} \rceil}{n}  &&\text{(since $L \leq 1$)} 
\nonumber 
\\ 
 &\begin{aligned}[c] \label{eq16}
 \leq  \begin{cases} 
      1- \frac{3}{2n} & \text{when $n$ is even} \\
      1- \frac{3}{2n} + \frac{1}{2n^2} & \text{when $n$ is odd}.
   \end{cases}
\end{aligned} 
\end{align}

From equation~\ref{eq15} and~\ref{eq16}, we have that the mechanism is $\epsilon$-strategyproof, where
\begin{align*}
 \begin{aligned}[c]
\epsilon = \begin{cases} 
      1- \frac{3}{2n} & \text{when $n$ is even} \\
      1- \frac{3}{2n} + \frac{1}{2n^2} & \text{when $n$ is odd}. 
   \end{cases}
\end{aligned}  
\end{align*} 
\end{proof} 

Although the modifiedEP was designed to take care of the worst-case scenarios that arise in the original algorithm, it turns out that it works only for the cases when $n=2$ 
or $n=3$. In the all the other cases, even though it still eliminates those worst-cases, it is interesting to note that it does this at the cost of introducing some other 
worst-case scenarios---ones which do not occur in the original algorithm---essentially because of the fact that it makes non-contiguous allocations (unlike the 
original Even-Paz which makes contiguous allocations). 
\section{Discussion}
One of the main open questions raised by \citet{chen10, chen13} (and subsequently also by \citet{aziz14} and \citet{branzei15}) was on the existence of deterministic 
mechanisms that are strategyproof and fair for piecewise constant valuations. Here we addressed a special case of this and we showed (in Theorem~\ref{SP-prop-cont}) that for 
any $n \geq 2$ agents there is no deterministic mechanism that makes contiguous allocations and is strategyproof and even $\epsilon$-proportional. Although 
the contiguous allocations constraint captures some interesting scenarios, we believe that answering the above question without this (rather strong) constraint is 
important. In particular, when given the fact that there exists randomized mechanisms that are a) truthful in expectation, proportional, and envy-free for piecewise linear 
valuations in the direct-revelation model \citep{chen13} and b) truthful in expectation and $\epsilon$-proportional for any $\epsilon > 0$ and for general valuation functions 
in the Robertson-Webb model \citep[Theorem 3]{branzei15}, we believe that resolving this question can give us a clearer picture regarding the limits of determinism.

Moving away from the above result, in Section~\ref{wm} we showed (in Theorem~\ref{ep-SP-gen}) that there exists a proportional mechanism that has better 
incentive-compatible properties 
than the Even-Paz algorithm for the cases when there are two or three agents. With regards to this result, one of the main questions that arise is the following and we 
suggest it as an open problem: for $n \geq 2$ agents, what is the minimum achievable $\epsilon$ for which there exists a deterministic mechanism that is proportional and 
$\epsilon$-strategyproof? It is  unclear how we can provide a lower bound on mechanisms that are proportional and $\epsilon$-strategyproof, or improve the proposed mechanism 
to get smaller values of 
$\epsilon$. And also, at the same time---and this relates back to our results in Section~\ref{nwm}---there is the question of whether we can make Algorithm~\ref{algo3} 
non-wasteful when agents have concisely representable valuation functions (as in, functions like piecewise uniform, piecewise constant 
etc.). At first glance it seems like it should be possible to easily do this at least when there are only two agents and one possible scheme that seems benign is: once 
the allocation is done as per the above mechanism, ask the agents to exchange their zero pieces (the pieces they value at zero)---thus resulting in an outcome that is possibly 
a Pareto-improvement over the existing one. Unfortunately it only \textit{seems} benign, but actually isn't as it turns out that it can result in breaking the guarantee on 
strategyproofness. We illustrate this through the following example. 

\begin{example} 
 The main issue with making the mechanism non-wasteful by doing as suggested above is that if the agents know that they will receive the zero pieces in the end, then 
they no longer have the incentives to behave like the way they did in the original mechanism. For instance, consider the case of two agents, where 
 \begin{equation*}
 \begin{aligned}[c]
v_1(x) = \begin{cases} 
      0 & x \in [0, 0.5] \\
      2  & x \in [0.5, 1].
   \end{cases}
\end{aligned}
\quad 
 \begin{aligned}[c]
v_2(x) = \begin{cases} 
      1 & x \in [0, 0.5] \\
      0  & x \in [0.5, 0.8]\\
      2.5 & x \in [0.8, 1].
   \end{cases}
\end{aligned}
\quad 
 \begin{aligned}[c]
v'_2(x) = \begin{cases} 
      0.4 & x \in [0, 0.5] \\
      1 & x \in [0.5, 0.8]\\
      2.5 & x \in [0.8, 1].
   \end{cases}
\end{aligned}
\end{equation*}
Now, according to our original mechanism if the reports of the agents were $(v_1, v_2)$, then agent 1 would get $[0.75, 1]$ + some part of $[0.5, 0.75]$, and agent 2 gets 
$[0, 0.5]$ + some part of $[0.5, 0.75]$. And so the total utility for agent 2 is at most $0.5$. But now, once we add the reallocation stage to make the mechanism 
non-wasteful, agent 2 knows that during reallocation she is assured to get $[0, 0.5]$ and so she can misrepresent her utility function by reporting $v'_2(x)$. This in turn 
will result in agent 1 getting $[0, 0.75]$ + some part of $[0.75, 0.8]$, and agent 2 getting $[0.8, 1] +$ some part of $[0.75, 0.8]$. However, in the reallocation stage agent 
2 will also be given $[0, 0.5]$, thus getting a utility of $1$---which makes it a $1/2$-SP mechanism.
\end{example}

Using a similar (but slightly more involved) example it can be shown that another seemingly promising scheme to make it non-wasteful where we was ask the players to exchange 
the zero pieces before the start of the mechanism and then run the mechanism on the rest of the cake does not work as well. So given the apparent difficulty and our lower 
bound (Theorem~\ref{nwImp}) in Section~\ref{nwm}, another problem that remains open is whether we can guarantee non-wastefulness along with proportionality and 
$\epsilon$-strategyproofness, where $\frac{1}{3n} \leq \epsilon < (1- \frac{1}{n})$.

\section*{Acknowledgments}
We thank Simina Br{\^a}nzei for useful discussions. 

\bibliographystyle{named}
\bibliography{paperDraft}   
  
\appendix 
\newpage
\section{Even-Paz algorithm} \label{appendixB}
\begin{algorithm}[bht]
\begin{center}
 \noindent\fbox{%
\begin{varwidth}{\dimexpr\linewidth-3\fboxsep-3\fboxrule\relax}
\begin{algorithmic}[1]
  \Procedurename EP$\left([a, b], S\right)$
   \Input $\left([a, b], S\right)$, where $[a, b] \subseteq [0, 1]$ is the cake to be proportionally allocated among $|S|$ agents 
  \State $k \leftarrow |S|$
  \If {$k == 1$} 
    \State \textbf{return} $([a, b])$
   \EndIf

  \For{each $i \in S$ }
    \State $c_i$ = \textbf{cut}$\left(i, a, \frac{\left\lfloor \frac{k}{2} \right\rfloor}{k} \cdot v_i(a, b)\right)$
  \EndFor
  
  \State Let $\{d_1, \cdots, d_{k} \}$ be the sorted order of the $c_i$'s (cutpoints).   $S_L \leftarrow \left\{i \in S \mid c_i \leq d_{\left\lfloor \frac{k}{2} 
\right\rfloor}\right\}$ and  $S_R \leftarrow \left\{i \in S \mid c_i > d_{\left\lfloor \frac{k}{2} \right\rfloor}\right\}$.

 \State Recursively call EP on $\left[a, d_{\left\lfloor \frac{k}{2} \right\rfloor}\right]$ and $\left(d_{\left\lfloor \frac{k}{2} \right\rfloor}, 
b\right]$ with $S_L$ and $S_R$, respectively. Return the allocation formed by these recursive calls.
\Main
\State output the allocation returned by EP$\left([0, 1], S = \{1, \cdots, n\}\right)$
 \end{algorithmic}  
\end{varwidth}
 }
\end{center}
\caption{Even-Paz algorithm. }
\label{algoEP}
\end{algorithm}

\end{document}